\documentclass{article}

\usepackage{amsmath}
\usepackage{amssymb}
\usepackage{epigraph}
\usepackage{MnSymbol}
\usepackage[amsmath,thmmarks]{ntheorem}

\title{\bf A Formal Sociologic Study of Free Will}
\author{%
  Giovanni Giuffrida\\ {\it University of Catania}\\ \texttt{giovanni.giuffrida@dmi.unict.it}
  \and 
  Calogero G. Zarba \\ {\it Neodata Intelligence} \\ \texttt{calogero.zarba@neodatagroup.com}
}
\date{\sc Friday, May 10, 2013}

\theoremstyle{plain}
\theorembodyfont{\upshape}
\theoremsymbol{\rule{1ex}{1ex}}
\theoremseparator{.}
\newtheorem{axiom}{Axiom}
\newtheorem{definition}{Definition}
\newtheorem{theorem}{Theorem}
\theoremheaderfont{\sc}
\theorembodyfont{\upshape}
\theoremstyle{nonumberplain}
\theoremseparator{.}
\theoremsymbol{\rule{1ex}{1ex}}
\newtheorem{proof}{Proof}

\DeclareMathOperator{\PROB}{\mathrm{Prob}}
\newcommand{\bbN}{\mathbb{N}}
\newcommand{\bbR}{\mathbb{R}}
\newcommand{\calC}{\mathcal{C}}
\newcommand{\calG}{\mathcal{G}}
\newcommand{\calP}{\mathcal{P}}

\begin{document}
\maketitle

\begin{abstract}
We make a formal sociologic study of the concept of free will.  By
using the language of mathematics and logic, we define what we call
everlasting societies.  Everlasting societies never age: persons never
age, and the goods of the society are indestructible.  The infinite
history of an everlasting society unfolds by following deterministic
and probabilistic laws that do their best to satisfy the free will of
all the persons of the society.

We define three possible kinds of histories for everlasting societies:
primitive histories, good histories, and golden histories.

In primitive histories, persons are inherently selfish, and they use
their free will to obtain the personal ownerships of all the goods of
the society.

In good histories, persons are inherently good, and they use their
free will to distribute the goods of the society.  In good histories,
a person is not only able to desire the personal ownership of goods,
but is also able to desire that a good be owned by another person.

In golden histories, free will is bound by the ethic of reciprocity,
which states that ``you should wish upon others as you would like
others to wish upon yourself''.  In golden societies, the ethic of
reciprocity becomes a law that partially binds free will, and that
must be abided at all times.  In other words, the verb ``should''
becomes the verb ``must''.
\end{abstract}

% --------------------------------------------------------------------
\tableofcontents

% --------------------------------------------------------------------
\epigraph{I have free will, but not of my own choice. I have never
  freely chosen to have free will. I have to have free will, whether I
  like it or not!}{\emph{Raymond Smullyan}}

% --------------------------------------------------------------------
\section{Introduction}
``Will'' is an English noun that is a synonym of ``desire''.  ``Free''
is an English adjective that means ``not subject to restrictions''.
Thus, the phrase ``free will'' means ``desire not subject to
restrictions''.  In this paper, we study the concept of free will by
defining a mathematical model of a society in which persons have the
ability of free will.

For us, a society comprises a set of persons, as well as a set of
goods.  At each instant of time, a given good of the society must be
owned by exactly one person.  Thus, for each instant of time, the
society has a social assignment that defines, for each good, the owner
of that good for the given instant of time.  

The kind of society that we study is everlasting, in the sense that
persons never age, and the goods of the society are indestructible.
Therefore, by modelling time as the set of natural numbers, the
infinite history of an everlasting society can be seen as an infinite
sequence of social assignments.

We make the assumption that the persons of the society have free will.
Persons have the ability to express desires over the ownership of the
goods of the society.  Since these desires are free, it is impossible
to avoid non-conflicting desires, and therefore there is an eternal
battle of wills for the ownership of the goods of the society.  At
each instant of time, and for each given good, there is a battle for
the ownership of the given good.  The outcome of this battle is
decided by probabilistic and deterministic laws that take into account
the free will of the persons of the society.  Thus, the free will of
the persons of the society can be seen as a physical force that acts
on the goods of the society.

In our mathematical model, the free will of a person is not entirely
free, but is bound by the power of the person.  The power of a person
is a measure of the ability of the person to generate desires, and
changes over time according to the laws of our mathematical model.

We define three possible kinds of histories for everlasting societies:
primitive histories, good histories, and golden histories.

In primitive histories, persons are inherently selfish, and they use
their free will to obtain the ownership of all the goods of the
society.  When a persons wins a battle over a given good, it is
entitled to own the good for one unit of time, but it does so at the
expense of losing some of her power.  When a persons loses a battle
over a given good, it will not own the good, but gains some power.  In
other words, there is a price to pay for winning a battle, and there
is something to gain when losing a battle.

In good histories, persons are inherently good, and they use their
free will to distribute the goods of the society.  In good histories,
a person is not only able to desire the personal ownership of goods,
but is also able to desire that a good be owned by another person.  In
good histories, a person wins a battle over a good when it is able to
determine the owner of the good, but if a person wins, it does so at
the expense of losing some of her power.  Symmetrically, a persons
loses a battle over a good when it is not able to determine the owner
of the good, but when a person loses, it gains some power.

In golden histories, like in the case of good histories, persons are
inherently good, and they are able to desire both the personal
ownership of goods, as well as the fact that a given good be owned by
another person.  But in golden histories free will is partially bound
by a law of reciprocity.  This law of reciprocity is inspired by the
golden rule, also known as the ethic of reciprocity, which states that
\emph{you should wish upon others as you would like others to wish
  upon yourself}.  In our mathematical model, the ethic of reciprocity
becomes a law that partially binds free will, and must be abided at
all times.  This law can be expressed by the two following statements:
\begin{enumerate}
\item \emph{You must wish in the future upon others, as others in the
  past have wished upon you};

\item \emph{Others must wish in the future upon you, as you in the
  past have wished upon others}.
\end{enumerate}

Our law of reciprocity can be expressed formally with mathematical
language, and can be proved valid in our mathematical model of
everlasting societies.

We do not know if our law of reciprocity is valid in the real world.
We do not even know if free will exists in the real world.  We however
believe that, in the real world, there are universal laws that govern
the concept of desire.  Our mathematical model can be seen as a
desire: it is the desire that free will exists, and that free will be
partially bound by reciprocity.

% --------------------------------------------------------------------
\subsection{Related work}
Free will~\cite{kan2005} has been studied by humanity since at least
the beginning of civilization.  Among the many possible views and
perspectives in the free will debate, the view of this paper can be
classified as a two-stage model~\cite{Jam1884}.  In the first stage,
alternative possibilities are generated.  In the second stage, exactly
one of the possibilities is chosen by an intelligent person capable of
free will.

The internal structure of everlasting societies can be seen as a
social network.  The theory of social networks is part of the emerging
field of computational sociology~\cite{BL2012, GT2005}, a field that
merges classical sociology~\cite{CRT2006} with computer
science~\cite{Tuc2004}.

% --------------------------------------------------------------------
\subsection{Organization of the paper}
The rest of this paper is organized as follows.
Section~\ref{sec:everlasting} defines everlasting societies.
Section~\ref{sec:primitive} defines primitive histories.
Section~\ref{sec:good} defines good histories.
Section~\ref{sec:golden} defines golden histories, and proves the
theorems that express mathematically the validity of the law of
reciprocity for golden histories.  Section~\ref{sec:conclusion}
concludes the paper.

% --------------------------------------------------------------------
\section{Everlasting societies}
\label{sec:everlasting}
\begin{axiom}
There are \emph{persons}.  Denote with $\calP$ the class of all
persons.
\end{axiom}

\begin{definition}
A \textsc{relationship function} over a set $P$ of persons is a function
\begin{equation*}
  \rho : P \times P \to ]0, 1] \,,
\end{equation*}
satisfying the following conditions:
\begin{itemize}
\item $\rho(x, x) = 1$, for all $x \in P$;

\item if $x \neq y$ then $0 < \rho(x, y) < 1$, for all $x, y \in P$;

\item $\rho(x, y) = \rho(y, x)$, for all $x, y \in P$;

\item $\rho(x, z) + \rho(z, y) \le 1 + \rho(x, y)$, for all $x, y, z
  \in P$.
\end{itemize}
\end{definition}

\begin{axiom}
There are \emph{goods}.  Denote with $\calG$ the class of all goods.
\end{axiom}

\begin{definition}
A \textsc{social estate} $E$ is a finite nonempty set of goods.
\end{definition}

\begin{definition}
An \textsc{everlasting society} $S$ is a tuple
\begin{equation*}
  S = (P, \rho, E) \,,
\end{equation*}
where
\begin{itemize}
\item $P \subseteq \calP$ is a finite nonempty set of persons;

\item $\rho: P \times P \to ]0, 1]$ is a relationship function over
    $P$;

\item $E$ is a social estate.
\end{itemize}
\end{definition}

\begin{definition}
Let $S = (P, \rho, E)$ be an everlasting society.  A \textsc{social
  assignment} $\alpha$ for $S$ is any function $\alpha: E \to P$.
\end{definition}

\begin{definition}
Let $S = (P, \rho, E)$ be an everlasting society.  A \textsc{simple
  history} for $S$ is any infinite sequence
\begin{equation*}
  \alpha_0, \alpha_1, \ldots, \alpha_n, \ldots \,,
\end{equation*}
of social assignments for $S$.
\end{definition}

% --------------------------------------------------------------------
\section{Primitive histories}
\label{sec:primitive}
\begin{definition}
Let $S = (P, \rho, E)$ be an everlasting society.  A \textsc{primitive
  power function} for $S$ is any function $\pi: P \to \bbR^+$.
\end{definition}

\begin{definition}
Let $S = (P, \rho, E)$ be an everlasting society.  A \textsc{primitive
  force function} for $S$ is any function $\varphi : P \times E \to
\bbR^+$.
\end{definition}

\begin{definition}
Let $S = (P, \rho, E)$ be an everlasting society.  A \textsc{primitive
  social state} $\sigma$ is a tuple
\begin{equation*}
  \sigma = (\alpha, \pi, \varphi) \,,
\end{equation*}
where
\begin{itemize}
\item $\alpha: E \to P$ is a social assignment for $S$;

\item $\pi: P \to \bbR^+$ is a primitive power function for $S$;

\item $\varphi: P \times E \to \bbR^+$ is a primitive force function
  for $S$ such that
\begin{align*}
  \left[ \displaystyle\sum_{a \in E} \varphi(x, a) \right] < \pi(x)
  \,, &&\text{for all } x \in P \,.
\end{align*}
\end{itemize}
\end{definition}

\begin{definition}
Let $S = (P, \rho, E)$ be an everlasting society, and let $\sigma =
(\alpha, \pi, \varphi)$ be a primitive social state.  The
\textsc{effectiveness function} $\psi$ with respect to $S$ and
$\sigma$ is the function
\begin{equation*}
  \psi : P \times E \to \bbR^+ \,,
\end{equation*}
defined by
\begin{equation*}
  \psi(x, a) = \varphi(x, a) \rho(x, \alpha(a)) \,.
\end{equation*}
\end{definition}

\begin{definition}
Let $S = (P, \rho, E)$ be an everlasting society.  For the sake of
simplicity, we assume in this definition that the estate of the
society contains only one good $a$, that is $E = \{a\}$.

Let $\sigma = (\alpha, \pi, \varphi)$ be a primitive social state.
Another primitive social state $\sigma' = (\alpha', \pi', \varphi')$
is a \textsc{successor} of $\sigma$, if it can be obtained according
to the following laws that regulate how the members of the society
battle for the ownership of good $a$.

The winner of this battle is established using a probabilistic law
that takes into account the force function $\varphi$, and in
particular the effectiveness function $\psi$ with respect to $S$ and
$\sigma$.

For any person $w \in P$, we let
\begin{equation*}
  \PROB [w \textrm{~wins the ownership of good~} a] = 
    \frac{\psi(w, a)}
         {\displaystyle\sum_{y \in P} \psi(y, a)} \,.
\end{equation*}

Next, assume that the battle has been won by person $w$.  Then the
successor state $\sigma' = (\alpha', \pi', \varphi')$ is uniquely
determined in $\alpha'$ and $\pi'$ (but not in $\varphi'$), according
to the following deterministic laws:
\begin{itemize}
\item $\alpha'(a) = w$;

\item $\pi'(w) = \pi(w) - \varphi(w, a)$;

\item if $x \neq w$ then
\begin{equation*}
  \pi'(x) = \pi(x) + \varphi(w, a) \times \frac{\psi(x, a)}{\displaystyle\sum_{\substack{y \in P \\ y \neq w}} \psi(y, a)} \,.
\end{equation*}
\end{itemize}
\end{definition}

\begin{definition}
Let $S = (P, \rho, E)$ be an everlasting society.  In this definition,
there is no simplicity, and we consider the general case in which $|E|
\ge 1$.

Let $\sigma = (\alpha, \pi, \varphi)$ be a primitive social state.
Another primitive social state $\sigma' = (\alpha', \pi', \varphi')$
is a \textsc{successor} of $\sigma$, if it can be obtained according
to the following laws that regulate how the members of the society
battle for the ownership of the goods belonging to the social estate
$E$.

There are as many battles as there are goods in $E$.  The winners of
these battles are established using a probabilistic law that takes
into account the force function $\varphi$, and in particular the
effectiveness function $\psi$ with respect to $S$ and $\sigma$.

For any person $w \in P$, and for any good $a \in E$, we let
\begin{equation*}
  \PROB [w \textrm{~wins the ownership of good~} a] = 
    \frac{\psi(w, a)}
         {\displaystyle\sum_{y \in P} \psi(y, a)} \,.
\end{equation*}

Once all battles have been settled using the probabilistic law, the
successor state $\sigma' = (\alpha', \pi', \varphi')$ is uniquely
determined in $\alpha'$ and $\pi'$ (but not in $\varphi'$) according
to the following deterministic laws.

First, we let $\alpha'(a)$ be equal to the winner of the battle over
good $a$.

Next, let $x \in P$.  Denote with $\mathit{gains}(x)$ the set of goods
that the person $x$ has gained.  More precisely, $\mathit{gains}(x) =
\{ a \in E \mid \alpha'(a) = x\}$.  Denote also with
$\mathit{losses}(x)$ the set of goods that the person $x$ has lost.
More precisely, $\mathit{losses}(x) = \{ a \in E \mid \alpha'(a) \neq
x\} = E \setminus \mathit{gains}(x)$.

Finally, for any person $x \in P$, we let
\begin{align*}
  \pi'(x) &= 
    \pi(x) \\
    &~~~ - \sum_{a \in \mathit{gains(x)}} \varphi(x, a) \\
    &~~~ + \sum_{a \in \mathit{losses(x)}} 
             \varphi(\alpha'(a), a) 
             \times 
             \frac{\psi(x, a)}
                  {\displaystyle\sum_{\substack{y \in P \\ y \neq \alpha'(a)}} \psi(y, a)} \,.
\end{align*}
\end{definition}

\begin{definition}
Let $S = (P, \rho, E)$ be an everlasting society.  A \textsc{primitive
  history} for $S$ is an infinite sequence of primitive social states
\begin{equation*}
  \sigma_0, \sigma_1, \ldots, \sigma_n, \ldots \,,
\end{equation*}
where $\sigma_{t+1}$ is a successor of $\sigma_t$, for all $t \in
\bbN$.
\end{definition}

% --------------------------------------------------------------------
\section{Good histories}
\label{sec:good}
\begin{definition}
Let $S = (P, \rho, E)$ be an everlasting society.  A \textsc{good
  power function} for $S$ is any function $\pi : P \times E \times P
\to \bbR^+$.
\end{definition}

\begin{definition}
Let $S = (P, \rho, E)$ be an everlasting society.  A \textsc{good
  force function} for $S$ is any function $\varphi : P \times E \times
P \to \bbR^+$.
\end{definition}

\begin{definition}
Let $S = (P, \rho, E)$ be an everlasting society.  A \textsc{good social
  state} $\sigma$ is a tuple
\begin{equation*}
  \sigma = (\alpha, \pi, \varphi) \,,
\end{equation*}
where
\begin{itemize}
\item $\alpha: E \to P$ is a social assignment for $S$;

\item $\pi : P \times E \times P \to ]0, 1[$ is a good power function
      for $S$;

\item $\varphi : P \times E \times P \to ]0, 1[$ is a good force
    function for $S$ such that
\begin{align*}
  \varphi(x, a, y) < \pi(x, a, y) \,,
  &&\text{for all } x, y \in P \text{ and } a \in E \,.
\end{align*}
\end{itemize}
\end{definition}

\begin{definition}
Let $S = (P, \rho, E)$ be an everlasting society, and let $\sigma =
(\alpha, \pi, \varphi)$ be a good social state.  The
\textsc{effectiveness function} $\psi$ with respect to the $S$ and
$\sigma$ is the function
\begin{equation*}
  \psi : P \times E \times P \to \bbR^+ \,,
\end{equation*}
defined by
\begin{equation*}
  \psi(x, a, y) = \varphi(x, a, y) \rho(x, \alpha(a)) \rho(\alpha(a), y) \,.
\end{equation*}
\end{definition}

\begin{definition}
Let $S = (P, \rho, E)$ be an everlasting society.  In this definition,
we consider the general case in which $|E| \ge 1$.

Let $\sigma = (\alpha, \pi, \varphi)$ be a good social state.  Another
good social state $\sigma' = (\alpha', \pi', \varphi')$ is a
\textsc{successor} of $\sigma$, if it can be obtained according to the
following laws that regulate how the members of the society battle
over the ownership of the goods belonging to the social estate $E$.

There are as many battles as there are goods in $E$.  The winners of
these battles are established using a probabilistic law that takes
into account the force function $\varphi$, and in particular the
effectiveness function $\psi$ with respect to $S$ and $\sigma$.

For any person $w \in P$, and for any good $a \in E$, we let
\begin{equation*}
  \PROB [w \textrm{~wins the ownership of good~} a] = 
    \frac{\displaystyle\sum_{y \in P} \psi(y, a, w)}
         {\displaystyle\sum_{y, z \in P} \psi(y, a, z)} \,.
\end{equation*}

Once all battles have been settled using the probabilistic law, the
successor state $\sigma = (\alpha', \pi', \varphi')$ is uniquely
determined in $\alpha'$ and $\pi'$ (but not in $\varphi'$), according
to the following deterministic laws.

First, we let $\alpha'(a)$ be equal to the winner of the battle over
good $a$.  Then, for any persons $x, y \in P$, and for any good $a \in
E$, we let
\begin{equation*}
  \pi'(x, a, y) = \pi(x) - \varphi(x, a, y) + \varphi(y, a, x) \,.
\end{equation*}
\end{definition}

\begin{definition}
Let $S = (P, \rho, E)$ be an everlasting society.  A \textsc{good
  history} for $S$ is an infinite sequence of good social states
\begin{equation*}
  \sigma_0, \sigma_1, \ldots, \sigma_n, \ldots \,,
\end{equation*}
where $\sigma_{t+1}$ is a successor of $\sigma_t$, for all $t \in \bbN$.
\end{definition}
  
% --------------------------------------------------------------------
\section{Golden histories}
\label{sec:golden}
\begin{axiom}
There are \emph{force carriers}.  Denote with $\calC$ the class of all
force carriers.
\end{axiom}

\begin{axiom}
Every force carrier $c$ has an \emph{intensity} $\mu(c) \in \bbR^+$.
\end{axiom}

\begin{axiom}
Every force carrier $c$ has a \emph{maximum idle period} $\theta(c)
\in \bbN^+$.
\end{axiom}

\begin{axiom}
Let $C \subseteq \calC$ be a finite nonempty set of force carriers.
An \textsc{idle function} for $C$ is any function $\tau: C \to \bbN$.
\end{axiom}

\begin{definition}
Let $S = (P, \rho, E)$ be an everlasting society.  A \textsc{golden
  power function} $\pi$ for $S$ is a function
\begin{equation*}
  \pi : P \times E \times P \to 2^C \,,
\end{equation*}
such that
\begin{itemize}
\item $C \subseteq \calC$ is a finite nonempty set of force carriers;

\item $\left[ \displaystyle\bigcup_{\substack{x, y \in P \\ a \in E}}
  \pi(x, a, y) \right] = C$;

\item if $(x_1, a_1, y_1) \neq (x_2, a_2, y_2)$ then $\pi(x_1, a_1,
  y_1) \cap \pi(x_2, a_2, y_2) = \emptyset$, for all $x_1, x_2, y_1,
  y_2 \in P$ and $a_1, a_2 \in E$.
\end{itemize}
\end{definition}

\begin{definition}
Let $S = (P, \rho, E)$ be an everlasting society, and let $\pi : P
\times E \times P \to 2^C$ be a golden power function for $S$.  The
\textsc{extended golden power function} $\pi^*$ with respect to $\pi$
is the function
\begin{equation*}
  \pi^* : 2^P \times E \times 2^P \to 2^C \,,
\end{equation*}
defined by
\begin{align*}
  \pi^*(X, a, Y) = \displaystyle\bigcup_{\substack{x \in X \\ y \in Y}} \pi(x, a, y) \,,
  &&\text{for all } X, Y \subseteq P \text{ and } a \in E \,.
\end{align*}
\end{definition}

\begin{definition}
Let $S = (P, \rho, E)$ be an everlasting society.  A \textsc{golden force
function} $\varphi$ for $S$ is a function
\begin{equation*}
  \varphi : P \times E \times P \to 2^C \,,
\end{equation*}
where $C \subseteq \calC$ is a finite nonempty set of force carriers.
\end{definition}

\begin{definition}
Let $S = (P, \rho, E)$ be an everlasting society, and let $\varphi : P
\times E \times P \to 2^C$ be a golden force function for $S$.  The
\textsc{extended golden force function} $\varphi^*$ with respect to
$\varphi$ is the function
\begin{equation*}
  \varphi^* : 2^P \times E \times 2^P \to 2^C \,,
\end{equation*}
defined by
\begin{align*}
  \varphi^*(X, a, Y) = \displaystyle\bigcup_{\substack{x \in X \\ y \in Y}} \varphi(x, a, y) \,,
  &&\text{for all } X, Y \subseteq P \text{ and } a \in E \,.
\end{align*}
\end{definition}

\begin{definition}
\label{def:golden_sigma}
Let $S = (P, \rho, E)$ be an everlasting society.  A \textsc{golden
  social state} $\sigma$ is a tuple
\begin{equation*}
  \sigma = (\alpha, \pi, \tau, \varphi)
\end{equation*}
where
\begin{itemize}
\item $\alpha: E \to P$ is a social assignment for $S$;

\item $\pi : P \times E \times P \to 2^C$ is a golden power function
  for $S$, where $C$ is a finite nonempty set of force carriers;

\item $\tau : C \to \bbN$ is an idle function for $C$;

\item $\tau(c) \le \theta(c)$, for all $c \in C$;

\item $\varphi : P \times E \times P \to 2^C$ is a golden force
  function such that
\begin{align*}
  \varphi(x, a, y) \subseteq \pi(x, a, y) \,, 
  &&\text{for all } x, y \in P \text{ and } a \in E \,;
\end{align*}

\item for any force carrier $c \in C$, if $\tau(c) = \theta(c)$ then
  $c \in \varphi(x, a, y)$, where $(x, a, y) \in P \times E \times P$
  is the unique tuple such that $c \in \pi(x, a, y)$.
\end{itemize}
\end{definition}

\begin{definition}
Let $S = (P, \rho, E)$ be an everlasting society, and let $\sigma =
(\alpha, \pi, \tau, \varphi)$ be a golden social state.  The
\textsc{effectiveness function} $\psi$ with respect to $S$ and
$\sigma$ is the function
\begin{equation*}
  \psi : P \times E \times P \to \bbR^+ \,,
\end{equation*}
defined by
\begin{align*}
  \psi(x, a, y) = \displaystyle\sum_{c \in \varphi(x, a, y)} \mu(c) \rho(x, \alpha(a)) \rho(\alpha(a), y) \,.
\end{align*}
\end{definition}

\begin{definition}
Let $S = (P, \rho, E)$ be an everlasting society, let $\sigma =
(\alpha, \pi, \tau, \varphi)$ be a golden social state, and let $\psi
: P \times E \times P \to \bbR^+$ be the effectiveness function with
respect to $S$ and $\sigma$.  The \textsc{extended effectiveness
  function} $\psi^*$ with respect to $\psi$ is the function
\begin{equation*}
  \psi^* : 2^P \times E \times 2^P \to \bbR^+
\end{equation*}
defined by
\begin{align*}
  \psi^*(X, a, Y) = \displaystyle\sum_{\substack{x \in X \\ y \in Y}} \psi(x, a, y) \,,
  &&\text{for all } X, Y \subseteq P \text{ and } a \in E \,.
\end{align*}
\end{definition}

\begin{definition}
\label{def:golden_succ}
Let $S = (P, \rho, E)$ be an everlasting society.  In this definition, 
we consider the general case in which $|E| \ge 1$.

Let $\sigma = (\alpha, \pi, \tau, \varphi)$ be a golden social state,
and let $\sigma' = (\alpha', \pi', \tau', \varphi')$ be another golden
social state.  Assume that $\tau$ and $\tau'$ have the same domain
$C$.  We say that $\sigma'$ is a \textsc{successor} of $\sigma$, if it
can be obtained according to the following laws that regulate how the
members of the society battle over the ownership of the goods
belonging to the social estate $E$.

There are as many battles as there are goods in $E$.  The winners of
these battles are established using a probabilistic law that takes
into account the force function $\varphi$, and in particular the
effectiveness function $\psi$ with respect to $S$ and $\sigma$.

For any person $w \in P$, and for any good $a \in E$, we let
\begin{equation*}
  \PROB [w \textrm{~wins the ownership of good~} a] = 
  \begin{cases}
    \frac{\psi^*(P, a, \{w\})} {\psi^*(P, a, P)} \,, &\text{if } \psi^*(P, a, P) \neq 0 \,, \\
    1                                            \,, &\text{if } \psi^*(P, a, P) = 0 \text{ and } w = \alpha(a) \,, \\ 
    0                                            \,, &\text{if } \psi^*(P, a, P) = 0 \text{ and } w \neq \alpha(a) \,. \\ 
  \end{cases}
\end{equation*}

Once all battles have been settled using the probabilistic law, the
successor state $\sigma = (\alpha', \pi', \tau', \varphi')$ is
uniquely determined in $\alpha'$, $\pi'$, and $\tau'$ (but not in
$\varphi')$ according to the following deterministic laws.

First, we let $\alpha'(a)$ be equal to the winner of the battle over
good $a$.  

Then, for any persons $x, y \in P$, and for any good $a \in E$, we let
\begin{equation*}
  \pi'(x, a, y) = 
  (\pi(x, a, y) \setminus \varphi(x, a, y)) \cup \varphi(y, a, x) \,.
\end{equation*}

Finaly, for any force carrier $c \in C$, let $(x, a, y) \in P \times E
\times P$ be the unique tuple such that $c \in \pi(x, a, y)$.  We let
\begin{align*}
  \tau'(c) = 
  \begin{cases}
    0           \,, &\text{if } c \in \varphi(x, a, y) \,, \\
    \tau(c) + 1 \,, &\text{if } c \notin \varphi(x, a, y) \,.
  \end{cases}
\end{align*}
\end{definition}

\begin{definition}
Let $S = (P, \rho, E)$ be an everlasting society.  A \textsc{golden
  history} for $S$ is an infinite sequence of golden social states
\begin{equation*}
  \sigma_0, \sigma_1, \ldots, \sigma_n, \ldots \,,
\end{equation*}
where $\sigma_{t+1}$ is a successor of $\sigma_t$, for all $t \in \bbN$.
\end{definition}

\begin{theorem}
\label{th:golden1}
Let $S = (P, \rho, E)$ be an everlasting society, and let
$\{\sigma_t\}_{t \in \bbN}$ be a golden history for $S$.  Let
$\sigma_{t_1} = (\alpha_{t_1}, \pi_{t_1}, \tau_{t_1}, \varphi_{t_1})$,
and assume $c \in \varphi_{t_1}(x, a, y)$, where $x, y \in P$ and $a
\in E$.

Then there exists $t_2 \in \bbN$ such that $t_2 > t_1$ and $c \in
\varphi_{t_2}(y, a, x)$.
\end{theorem}
\begin{proof}
Since $c \in \varphi_{t_1}(x, a, y)$, by
Definition~\ref{def:golden_succ}, we have $c \in \pi_{t_1 + 1}(y, a,
x)$ and $\tau_{t_1+1}(c) = 0$.  By contradiction, assume that $c
\notin \varphi_{t}(y, a, x)$, for all $t \in \bbN$ such that $t >
t_1$.  But then, by Definition~\ref{def:golden_succ}, there exists
$t_2 \in \bbN$ such that $t_2 > t_1$, $c \in \pi_{t_2}(y, a, x)$,
and $\tau_{t_2}(c) = \theta(c)$.  But then, by
Definition~\ref{def:golden_sigma}, we have $c \in \varphi_{t_2}(y, a, x)$,
contradiction.
\end{proof}

\begin{theorem}
\label{th:golden2}
Let $S = (P, \rho, E)$ be an everlasting society, and let
$\{\sigma_t\}_{t \in \bbN}$ be a golden history for $S$.  Let
$\sigma_{t_1} = (\alpha_{t_1}, \pi_{t_1}, \tau_{t_1}, \varphi_{t_1})$,
and assume $c \in \varphi^*_{t_1}(X, a, Y)$, where $X, Y \subseteq P$
and $a \in E$.

Then there exists $t_2 \in \bbN$ such that $t_2 > t_1$ and $c \in
\varphi^*_{t_2}(Y, a, X)$.
\end{theorem}
\begin{proof}
Since $c \in \varphi^*_{t_1}(X, a, Y)$, there exist $x, y \in P$ such
that $c \in \varphi_{t_1}(x, a, y)$.  By Theorem~\ref{th:golden1},
there exists $t_2 \in \bbN$ such that $t_2 > t_1$ and $c \in
\varphi_{t_2}(y, a, x)$.  But then, $c \in \varphi^*_{t_2}(Y, a, X)$.
\end{proof}

% --------------------------------------------------------------------
\section{Conclusion}
\label{sec:conclusion}
We have defined a mathematical model of everlasting societies in which
persons never age and goods are indestructible.  The history of an
everlasting society is governed by probabilistic and deterministic
laws that take into account the free will of the persons of the
society.

We have defined three possible kinds of histories for an everlasting
society: primitive histories, good histories, and golden histories.
The most interesting case is that of golden histories.  In golden
histories, free will is not entirely free, but is bound by a law of
reciprocity that is inspired by the golden rule, or ethic of
reciprocity.

Our law of reciprocity can be expressed by the two following
statements:
\begin{enumerate}
\item \emph{You must wish in the future upon others, as others in the
  past have wished upon you};

\item \emph{Others must wish in the future upon you, as you in the
  past have wished upon others}.
\end{enumerate}
In our mathematical model, statements 1 and 2 are expressed by
Theorem~\ref{th:golden2}.  Loosely speaking, the theorem states a
relationship between an instant of time $t_1$ in the past and an
instant of time $t_2$ in the future.  The theorem assumes that in the
past $c \in \varphi_{t_1}(X, a, Y)$, for some force carrier $c$, some
groups of persons $X, Y \subseteq P$, and some good $a$.  Given this
assumption, the theorem states that in the future we must have $c \in
\varphi_{t_2}(Y, a, X)$.  Intuitively, $\varphi_{t_1}(X, a, Y)$ is
what in the past $X$ has desired on $Y$, whereas $\varphi_{t_2}(Y, a,
X)$ is what in the future $Y$ must desire on $X$.  It follows that
Theorem~\ref{th:golden2} paraphrases statement~1 if we let $X = P
\setminus \{ \mathit{you} \}$ and $Y = \{ \mathit{you} \}$.
Similarly, Theorem~\ref{th:golden2} paraphrases statement~2 if we let
$X = \{ \mathit{you} \}$ and $Y = P \setminus \{ \mathit{you} \}$.

The laws of our mathematical model are good, in the sense that they
attempt to satisfy all desires of the persons of the society.
Moreover, in the case of golden histories, our laws are reciprocal.
However, our laws are not fair.

In our model, persons generate their desire by defining a force
function at each instant of time.  This force function is contrained
by the power function, and different persons have different power.
Moreover, the effectiveness of the force function depends on the
structure of the society, namely on the strength of the relationships
between persons.  Thus, the ability of a person to have their desires
satisfied, depends also on the social position that the person has in
the social network representing the society.

A direction of further research is that to define a mathematical model
that extends the mathematical model of this paper, and that is able to
define laws that are fair.  Here, by fairness we mean that persons
should start on equal footing, and that two distinct persons must have
the same chance to have their desires satisfied, provided only that
they behave intelligently.  

This paper, as stated in the introduction, is a desire.  It is the
desire that the real world be governed by laws that allow the
existence of free will, and that are fundamentally good, reciprocal,
and fair.  We do not know if the real world is good, reciprocal, and
fair.  We just desire that it be so.

%% It is our strong belief that the real world is governed by fundamental
%% laws that are inherently good, reciprocal, and fair.  The discovery
%% of such laws is undoubtedly the main quest of humanity.

% --------------------------------------------------------------------
\bibliography{wills}

\begin{thebibliography}{1}

\bibitem{BL2012}
Phillip Bonacich and Philip Lu.
\newblock {\em Introduction to Mathematical Sociology}.
\newblock Princeton University Press, 2012.

\bibitem{CRT2006}
Craig Calhoun, Chris Rojek, and Bryan~S. Turner, editors.
\newblock {\em The {SAGE} Handbook of Sociology}.
\newblock Sage Publications, 2006.

\bibitem{GT2005}
Nigel Gilbert and Klaus Troitzsch.
\newblock {\em Simulation for the Social Scientist}.
\newblock Open University Press, second edition, 2005.

\bibitem{Jam1884}
William James.
\newblock The dilemma of determinism.
\newblock {\em Unitarian Review}, 1884.

\bibitem{kan2005}
Robert Kane, editor.
\newblock {\em The Oxford Handbook of Free Will}.
\newblock Oxford University Press, 2005.

\bibitem{Tuc2004}
Allen~B. Tucker, editor.
\newblock {\em Computer Science Handbook}.
\newblock Chapman and Hall, second edition, 2004.

\end{thebibliography}
\end{document}